\documentclass[twocolumn,twoside]{IEEEtran}

\usepackage{graphicx}
\usepackage{wrapfig}
\usepackage{epsfig}
\usepackage{graphicx}
\usepackage{color}
\usepackage{mdframed}
\usepackage{mathrsfs} 
\usepackage{amsfonts}
\usepackage{latexsym}
\usepackage{amsmath}
\usepackage{amssymb}
\usepackage{color}
\usepackage{float} 
\usepackage{floatflt} 
\usepackage{caption}
\usepackage{subcaption}
\usepackage{balance}
\usepackage{cite}
\usepackage{pdfpages}

\newcommand{\suppress}[1]{}

\newtheorem{theorem}{Theorem}[section]
\newtheorem{lemma}{Lemma}[section]

\newtheorem{claim}{Claim}[section]

\newtheorem{definition}{Definition}[section]
\newtheorem{corollary}{Corollary}[section]
\newtheorem{remark}{Remark}[section]

\newcommand{\bRk}{{\bf R}}
\newcommand{\bRkk}{{\bf R}}

\newcommand{\bRs}{{\bf R}_{\tt SR}}

\newcommand{\rs}{{\tt SR}}
\newcommand{\rk}{ }

\def\proof{\par\penalty-1000\vskip .5 pt\noindent{\bf Proof\/: }}

\def\proof{\noindent\hspace{2em}{\it Proof: }}

\def\cA{\mbox{$\cal{A}$}}

\def\cX{\mbox{$\cal{X}$}}

\def\cB{\mbox{$\cal{B}$}}

\def\cK{\mbox{$\cal{K}$}}

\newcommand{\cI}{{{\cal I}}}

\newcommand{\bR}{{{\bf R}}}

\def\01{\{0,1\}}

\newcommand{\remove}[1]{}

\begin{document}

\title{Multiple Key-cast over Networks}

\author{Michael Langberg\ \ \ \ \ \ \ \ \  \ \ \ \ \ \ \ \ \ 
Michelle Effros
\thanks{M. Langberg is with the Department of Electrical Engineering at the University at Buffalo (State University of New York).  
Email: \texttt{mikel@buffalo.edu}}
\thanks{M. Effros is with the Department of Electrical Engineering at the California Institute of Technology.
Email: \texttt{effros@caltech.edu}}
\thanks{This work is supported in part by NSF grants CCF-1817241 and CCF-1909451. 
}
}


\maketitle

\begin{abstract}
The multicast key-dissemination problem over noiseless networks, introduced by Langberg and Effros [ITW 2022], here called the {\em key-cast} problem, captures the task of disseminating a shared secret random key to a set of terminals over a given network. 
Unlike traditional communication, where messages must be delivered from source to destination(s) unchanged, key-cast is more flexible since key-cast need not require source reconstruction at destination nodes.
For example, the distributed keys can be  mixtures of sources from which the sources themselves may be unrecoverable.

The work at hand considers key dissemination in the single-source, multiple-multicast network coding setting, i.e.,  the {\em multiple key-cast} problem.
Here, distinct keys are to be simultaneously transmitted from a single source node to multiple terminal sets, one shared random key per multicast set.
Scenarios include the secure setting, in which only the source and intended destinations gain information about a given key; and the non-secure setting in which the only requirement is that the knowledge of one key does not reveal information about another.
In both settings, we present combinatorial conditions for key dissemination and design corresponding multiple key-cast schemes.
In addition, we compare the multiple key-cast rate with and without the restriction of source reconstruction, the former corresponding to  traditional forms of communication; key-cast achieves a strict advantage in rate when source reconstruction is relaxed.
\end{abstract}

\thispagestyle{empty}

\section{Introduction} 
\label{sec:intro}
The resource of shared randomness plays a fundamental role in the theory and practice of network communication systems.
A uniformly distributed key, shared among some network users and, at times, hidden from others, appears as a central resource in a variety of communication tasks and is used, for example, in the secure transmission of information; in randomized coding techniques in the presence of uncertain noise models; in the context of distributed computing, statistical inference, and distributed learning, though the availability of public coins; and in distributed authentication, identification, and local differential privacy through shared forms of sampling and hashing, e.g., \cite{shannon1949communication,maurer1993secret,ahlswede1993common,lapidoth1998reliable,chen2021breaking,acharya2019communication}.
%
%

A {\em key-dissemination} communication protocol is one in which a key $K$, or a collection of keys $\cK=\{K_1,\dots,K_\ell\}$, which may be required to be kept secret, are shared among a collection of users as a prelude to future communication tasks requiring shared randomness.
The task of key dissemination (also called secret key-agreement) has seen significant studies over the past decades in the context of isolated network structures, e.g., \cite{gacs1973common,wyner1975common,wyner1975wire,witsenhausen1975sequences,csiszar1978broadcast,bennett1988privacy,ahlswede1993common,maurer1993secret,bennett1995generalized,maurer1997privacy,ahlswede1998common,csiszar2000common,csiszar2004secrecy,mossel2006non,chan2014multiterminal,csiszar2008secrecy,gohari2010information,gohari2010information2,siavoshani2010group,bogdanov2011extracting,tyagi2013common,chan2014extracting,liu2015secret,guruswami2016tight,xu2016private,hayashi2016secret,narayan2016multiterminal,ghazi2018resource,liu2016common,canonne2017communication} 
in which a collection of users wish to share a common key over a noisy network structure that is subject to eavesdropping. 

The problem of key-dissemination in the context of network coding (i.e., noiseless networks) was recently introduced in \cite{LE:22}.
Specifically,  \cite{LE:22} studies the multiple-source, single-multicast setting, here called the {\em key-cast} setting,  in which one wishes to multicast a uniform key $K$ of rate $R$ to a collection of terminal nodes. 
Sources have access to independent randomness, and, as the network is noiseless, the resulting key $K$ is a function of the sources' information.
Key dissemination in this context resembles the task of secure multicast network coding,  e.g., 
\cite{cai2002secure,feldman2004capacity,cai2007security,yeung2008optimality,el2012secure,silva2011universal,jaggi2012secure},  as information eventually shared between terminals is kept secret from the network eavesdropper.
However the two tasks differ in that the former is 
more flexible.
Specifically, in the latter, source nodes hold message information that must be reconstructed at terminal nodes while  
in the former source reconstruction is not required since keys can be mixtures of sources from which the sources themselves may be unrecoverable.
This flexibility in key-dissemination opens the possibility of a key-rate $R$ that exceeds that obtainable through secure-multicast using source reconstruction.
Indeed, for the key-cast setting, \cite{LE:22} shows a significant gap between the key rates achievable with and without the requirement of source reconstruction.

The work at hand continues this line of study, and addresses key dissemination in the single-source, multiple-multicast, network-coding setting, where there are multiple terminal sets, each requiring a distinct key.
We refer to this problem as the (single-source) {\em multiple key-cast} problem.
The simultaneous dissemination of distinct keys to distinct terminal sets over networks is useful
as a prelude to future communication tasks within each terminal set and for multiparty applications that require unique identification, authentication, and private communications obtained through key dissimenation, e.g.,
\cite{lim2005extracting,su2008digital,suh2007physical,yu2009towards,chen2020breaking,byrd2020differentially,ahlswede2021identification}.
We study multiple key-cast in both the secure and non-secure setting.

In the secure setting, we seek the dissemination of distinct keys to distinct terminal sets under the requirement that no individual network node other than the source $s$ and each intended terminal node learns any information regarding each of the keys. 
All concepts are described in full detail in Section~\ref{sec:model}.
Using the paradigm of secret sharing, e.g., \cite{shamir1979share,ito1989secret,beimel2011secret}, and inspired by the study of distributed secret sharing in the context of single-source network coding \cite{shah2013secure,shah2015distributed}, we present a tight sufficient combinatorial condition (and a corresponding communication scheme) for secure multiple key-cast.

In the non-secure case, we require both that each terminal set decode a distinct key and that the knowledge of one key does not reveal information on any other key disseminated through the network.
Inspired by the analyses appearing in \cite{wang2007intersession,wang2010pairwise,shenvi2010simple}, which  address coding solutions and upper bounds for $2$-unicast network coding, we present tight necessary and sufficient combinatorial conditions (and a corresponding communication scheme) for multiple key-cast.

Finally, to better understand the place of source reconstruction in the context of key dissemination protocols, we compare key-dissemination with the more traditional form of communication in which source information is first reconstructed at the terminals and only then (perhaps) post-processed to derive a shared key.
We show, for both the secure and non-secure case, a significant gap between the  key-rates obtainable with and without the requirement of source reconstruction.

The remainder of the work is structured as follows. 
A detailed model is given in Section~\ref{sec:model}.
Our main results are given in Section~\ref{sec:results}, first for the non-secure case and then for the secure case.
In Section~\ref{sec:compare} we present both secure and non-secure instances, corresponding to the analysis of Section~\ref{sec:results}, for which there are significant differences between the key-rates obtainable with and without the requirement of source reconstruction.
Several technical proofs are deferred to the Appendix.


\section{Model} 
\label{sec:model}


We follow the model and definitions given in \cite{LE:22}, with slight modifications to fit the problems studied in this work.
The following notation is useful to the definitions that follow. For any integer $\ell$ let $[\ell]=\{1,2,\dots,\ell\}$.
\vspace{2mm}

\noindent
{\bf $\bullet$ Acyclic Multiple Key-cast Instance:} An instance ${\mathcal I}=(G,s,\{D_i\}_{i=1}^\ell,\{\cB_i\}_{i=1}^\ell)$ of the multiple key-cast problem includes an acyclic
directed network $G=(V,E)$ in which each edge $e \in E$ has unit capacity (we allow multiple parallel edges to capture connectivity of higher integer capacity), a  source node $s \in V$, a collection of disjoint terminal sets $D_i \subseteq V$ for $i \in [\ell]$, each consisting of a collection of terminal nodes, and for $i \in [\ell]$ a collection of subsets of edges $\cB_i=\{\beta_{i,1},\dots,\beta_{i,{\tiny |\cB_i|}}\}$ specifying the secrecy requirements.
Source $s$ holds an unlimited collection of independent uniformly-distributed bits $M=\{b_{i}\}_i$.  
Following a convention that is common in the study of acyclic network coding, we assume that $s$ has no incoming edges, and that terminals $d \in \cup_{i \in [\ell]}D_i$ have no outgoing edges. 

\vspace{2mm}

\noindent
{\bf $\bullet$ Key Codes:}
For blocklegth $n$, network code 
$({\mathcal F},\mathcal{G})=(\{f_{e}\},\{g_{i,j}\})$
is an assignment of a (local) encoding function
$f_{e}$ for each edge $e\in E$ and a decoding function $g_{i,j}$ for each terminal $d_{i,j} \in D_i$, for $i \in [\ell]$. 
For every edge $e=(u,v)$, the edge message
$X^n_{e} \in \cX^n_{e}=[2^{n}]$ from $u$ to $v$ equals the evaluation of encoding function $f_{e}$  on inputs $X^n_{{\rm In}(u)}$;
where, for a generic node $u_0$, $X^n_{{\rm In}(u_0)}$ equals
$((X^n_{e'}:e' = (v,u_0) \in E), (M: u_0=s))$ 
 captures all information available to node $u_0$ during the communication process.
In order to ensure that $X^n_{{\rm In}(u)}$ is available to node $u$ before it encodes, communication proceeds according to a predetermined topological order on $E$.
A key code with target rate $R$ is considered successful if for each $i \in [\ell]$ and every terminal $d_{i,j} \in D_i$ the evaluation of decoding functions $g_{i,j}$ on the vector of random variables $X^n_{{\rm In}(d_{i,j})}$ equals the reproduction of a uniform random variable $K_i$ over alphabet $[2^{Rn}]$ such that the following criteria are satisfied.
First, key $K_i$ meets secrecy constraints $\cB_i$, which specifies that for every $\beta \in \cB_i$, $I(K_i;(X^n_e: e \in \beta))=0$.
Second, each terminal-set $D_i$ wants a distinct key $K_i$ such that for $i \ne i'$ key $K_i$ is independent of key $K_{i'}$, i.e., random variables $\{K_i\}_{i=1}^\ell$ are pair-wise independent (PWI), giving $I(K_i;K_{i'})=0$ for all $i \ne i'$.

\begin{definition}[Multiple key-cast feasibility]
\label{def:key}
Instance $\cI$ is said to be $(R,n)$-feasible if there exists a key code $({\mathcal F},\mathcal{G})$ with blocklength $n$ such that
\begin{itemize}
	\item {\bf Decoding:} For all $i \in [\ell]$ and all $d_{i,j} \in D_i$, $H(K_i|X^n_{{\rm In}(d_{i,j})})=0$.
	\item {\bf PWI key rate:} For all $i \in [\ell]$, $K_i$ is a uniform
	random variable with $H(K_i)=Rn$. For $i \ne i' \in \{1,\dots,\ell\}$, $I(K_i;K_{i'})=0$.
	\item {\bf Secrecy:}  For all $i \in [\ell]$, $I(K_i;(X^n_{e}:e \in \beta))=0$ for any subset $\beta \in \cB_i$.
\end{itemize}
\end{definition}

In this study, we also compare key dissemination with more traditional forms of communication in which source information is first reconstructed at the terminals and only then (perhaps) post-processed to derive a shared key $K_i$, we call that approach source-reconstructed (SR) key-dissemination.
\begin{definition}[Source-reconstructed multiple key-cast feasibility]
\label{def:trad}
Instance $\cI$ is said to be $(R,n)_\rs$-feasible if there exists a key code $({\mathcal F},\mathcal{G})$ with blocklength $n$ such that
\begin{itemize}
	\item {\bf Source reconstruction:} For all $i \in [\ell]$ and all $d_{i,j} \in D_i$, there exists a collection of source information bits $M_{i,j} \subseteq M$, such that $H(M_{i,j}|X^n_{{\rm In}(d_{i,j})})=0$, i.e., message bits in $M_{i,j}$ are decoded at terminal $d_{i,j}$.
	\item {\bf PWI key construction and rate (post-processing):} For all $i \in [\ell]$ there exists a uniform random variable $K_i$ with $H(K_i)=Rn$ such that for all $d_{i,j} \in D_i$, $H(K_i|M_{i,j}) = 0$. For $i \ne i' \in \{1,\dots,\ell\}$, $I(K_i;K_{i'})=0$.
	\item {\bf Secrecy:}  For all $i \in [\ell]$, $I(K_i;(X^n_{e}:e \in \beta))=0$ for any subset $\beta \in \cB_i$.
\end{itemize}
\end{definition}

\begin{definition}
[Multiple key-cast capacity]
\label{def:cap_k}
The (symmetric) multiple key-cast capacity of $\cI$, denoted by $\bRk(\cI)$, is the maximum $R$ for which for all ${\Delta}>0$ there exist infinitely many blocklengths $n$ such that $\cI$ is $(R-{\Delta},n)$-feasible. 
The capacity obtainable by first reconstructing source information and then post-processing the shared key is denoted by $\bRs(\cI)$ and is defined analogously.
\end{definition}

\section{Results}
\label{sec:results}

\subsection{Non-secure case}

In this section we present a protocol for multiple key-cast in the non-secure setting and compare the achievable key-rate with that of traditional protocols that reconstruct source information.
We study the problem of distributing a key to all members of a given terminal set, with distinct keys simultaneously going to distinct terminal sets. 
The distributed keys should have the property that the key held by any terminal node does not reveal any information about any key shared by terminals from a different terminal-set. 
That is, we require the keys to be pair-wise independent.


Our protocol and analysis are inspired by the analysis appearing in \cite{wang2007intersession,wang2010pairwise,shenvi2010simple},  which address coding solutions and upper bounds for $2$-unicast network coding with integral edge capacities.
Roughly speaking, we here show, for single-source instances with multiple terminal sets, that unit-rate, multiple key-cast  is possible if and only if for every terminal $d_i$ in terminal set $D_i$, and for every $j \ne i$, there exists a unit-capacity path from the source $s$ to $d_i$ that does not pass through certain cut-sets $C_j$ corresponding to terminals $d_j \in D_j$. 
The converse follows standard cut-set arguments while the achievability argument combines a two-phase process in which an initial 2-multicast linear network code is modified to guarantee, for each $j$, that all terminals in $D_j$ decode the same key $K_j$ and that $K_j$ is independent from the key decoded by any other terminal set $D_i$ for $i \ne j$. 

To specify the cut-sets $C_j$ corresponding to terminals $d_j \in D_j$ we use the following definition from \cite{wang2007intersession,wang2010pairwise,shenvi2010simple}.
Throughout, we assume a predetermined topological order on the edges of $G$.

\begin{definition}[The cut sets $C_j$]
\label{def:cutsets}
For every $j \in [\ell]$ and $d \in D_j$ for which there exist one or more edges whose removal separates $d$ from the source $s$, let $e_d$ be the separating edge of minimum topological order in $G$; otherwise, let $e_d = \phi$.
For every $j \in [\ell]$, let $C_j = \{e_d | d \in D_j, \ e_d \ne \phi\}$.
\end{definition}

The main theorem of this section suggests a combinatorial characterization for key dissemination at unit rate. Proof is given in Appendix~\ref{sec:app_non}. A rough proof outline of the achievability scheme follows below.

\begin{theorem}[Multiple key-cast]
\label{the:kd}
Consider an instance  ${\mathcal I}=(G,s,\{D_i\}_{i=1}^\ell,\{\cB_i\}_{i=1}^\ell)$  of the multiple key-cast problem with $\cB_i=\phi$ for $i \in [\ell]$ (i.e., with no security constraints).
Then $\bR(\cI) \geq 1$ if and only if for every $i, j \in [\ell]$ such that $j \ne i$ and for every terminal $d \in D_i$ there exists a unit-capacity path connecting $s$ to $d$ that does not use edges in $C_j$.  
\end{theorem}

In Section~\ref{sec:compare}, we compare the achievable key rate $\bR(\cI)$ of our scheme with the maximum key-rate $\bRs(\cI)$ obtainable through source reconstruction and show a significant gap. Namely, we prove the following theorem.

\begin{theorem}[Multiple key-cast with source reconstruction]
\label{the:gap1}
Let $\epsilon >0$.
There exist instances $\cI$ of the non-secure, multiple key-cast problem that satisfy the sufficient conditions of Theorem~\ref{the:kd} for which $\bRs(\cI) \leq 3/4 + \epsilon$.
\end{theorem}

\subsection{Proof of  Theorem~\ref{the:kd}, rough outline of achievability}
For achievability, we design a two-stage encoding scheme; both stages are deterministic. First, we design a 2-multicast coding solution using a certain edge-coloring of $G$. Then, the coloring and coding scheme are modified to match our key dissemination requirements. 
Describing 2-multicast through edge-coloring is used, e.g., in \cite{fragouli2007network}.

Let the source $s$ hold 2 messages $a$ and $b$. 
In our edge-colorings, an edge $e$ colored by the color $\alpha$ represents the transmission of the linear combination $a+\alpha b$ on $e$, where $a$, $b$, and $\alpha$ are all elements of a sufficiently large field $F=[2^n]$ for blocklength $n$, and all operations are done over $F$.
Our coloring is governed by the predetermined topological order of edges in $G$.
We assume, without loss of generality, that every node in $G$ is connected from $s$.
Otherwise, one can remove such nodes from $G$ without impacting the communication protocol.

\noindent
{\bf $\bullet$ The first coloring stage:}
Consider the edge $e$ of least topological order.
Let $T_e$ be the set of edges that are disconnected from $s$ by the removal of $e$ (we call such edges $e$-tight).
We color $e$ and every $e' \in T_e$ with the color $\alpha=1$ corresponding to the message $a+\alpha b = a+b$.
Notice that the coding scheme that transmits $a+\alpha b$ for $\alpha=1$ on $e$ and on all edges stemming from $e$ in $T_e$ is a valid key code in the sense that the incoming information to any edge suffices to compute its outgoing information.
Next, we continue coloring by induction over the topological order of edges $e$ in $G$, where in each step we consider the next uncolored edge $e$ in topological order and color $e$ and the corresponding set $T_e$ (of edges disconnected from $s$ by the removal of $e$) by a new color $\alpha$, greater (by one) than all previous colors assigned, corresponding to the message $a + \alpha b$ to be communicated on $e$. 
In Appendix~\ref{sec:app_non} we prove that in any intermediate phase of our induction, any edge that has been assigned a color suffices to compute its outgoing information from its incoming information.
Our first coloring stage is depicted in Figure~\ref{fig:jj}(a). 

\noindent
{\bf $\bullet$ The second coloring stage:} To initiate our second coloring/coding stage, we now focus on the cut sets $C_j$ defined previously, and on the set of edges $e'$ that are disconnected from $s$ by the removal of $C_j$. 
We denote this latter set of edges by $T_j$, and refer to such edges as $j$-tight.
In Appendix~\ref{sec:app_non} we prove that any edge can be $j$-tight for at most one value of $j \in [\ell]$.
By the toplogical-minimality condition in the definition of edges $e \in C_j$ (Definition~\ref{def:cutsets}), it holds that $e$ is either 2-edge connected from $s$ or is an outgoing edge of $s$.
This implies, for any $j$ and $j'$ ($j$ may equal $j'$),  that distinct edges $e \in C_j$ and $e' \in C_{j'}$ will have distinct colors.
In our second stage of coloring, we assign for all $j \in [\ell]$, a new color $\alpha_j$ to all edges $e$ in $C_j$ and to all $j$-tight edges $e' \in T_j$.
As before, coloring an edge by $\alpha_j$ corresponds to the transmission of $a+\alpha_j b$ on that edge.
In Appendix~\ref{sec:app_non} we prove that this modified coloring does not impact the network coding feasibility.
Namely, any edge $e$ can compute its outgoing information from its incoming information.

\noindent
{\bf $\bullet$ The decoding of $K_j=a + \alpha_j b$ at terminals $d \in D_j$:}
To finish our proof, we need to show that for any $j \in [\ell]$, any terminal $d \in D_j$ is able to decode $K_j = a+\alpha_j b$ (of rate 1).
Notice that this collection of keys is pair-wise independent.
We here assume, without loss of generality, that all terminal nodes $d$ in $G$ have only one incoming edge. 
Otherwise, for any terminal $d \in D_j$ one can construct a new instance by adding to $G$ a new node $d'$, adding a new edge $(d,d')$, and modifying $D_j$ by removing $d$ and adding $d'$.
The new instance is solvable at rate 1 if and only if the original instance is solvable at rate 1.  
With this assumption, the single edge $e$ incoming  to $d \in D_j$ is either in $C_j$ or in $T_j$.
This follows from the observation that $d$ is separated from $s$ by the removal of its single incoming edge, and thus there exists an edge $e_{d} \in C_j$ of minimum topological order disconnecting $d$ from $s$.
As edges in $C_j$ and $T_j$ are colored by $\alpha_j$,  terminal $d \in D_j$ can decode $K_j = a+\alpha_j b$.
This concludes the rough outline of our achievability proof. Full details appear in Appendix~\ref{sec:app_non}.

\section{Secure case}

In this section, we consider (single-source) multiple key-cast in which one distributes a collection of keys $\cK=\{K_1,\dots,K_\ell\}$ to the terminals in disjoint terminal sets $\{D_1,\dots,D_\ell\}$ under the requirement that 
for each $j \in [\ell]$, the only network nodes $v \in V \setminus \{s\}$ that individually hold any information regarding key $K_j$ are the nodes $v \in D_j$. 
We study the key capacity in this setting through the lens of {\em secret sharing}.

In the secret sharing paradigm (e.g., \cite{shamir1979share,ito1989secret,beimel2011secret}) a {\em dealer}, who holds a uniformly distributed secret message, is required to distribute {\em shares} to a collection of {\em users}, giving each user one share.
Each share is a random variable computed by the dealer using the secret message and additional randomness.
An {\em access structure} ($\cA_{\tt access},\cA_{\tt no-access})$ is a predetermined collection of subsets of users, such that each subset
of users in $\cA_{\tt access}$ can jointly decode the secret and each subset of users in $\cA_{\tt no-access}$ learns nothing about the secret through attempts at joint decoding. 
For example, {\em threshold} access structures \cite{blakley1979safeguarding,shamir1979share} require that any subset of $k$ users cannot learn anything about the secret message and that any collection of $k+1$ users can jointly recover the secret.

Inspired by \cite{shah2013secure,shah2015distributed}, which design secret sharing protocols over network structures, we present  combinatorial conditions  allowing 
single-source, multiple key-cast under the security requirements specified above.

\begin{theorem}[Secure multiple key-cast]
\label{the:skd}
Consider an instance  ${\mathcal I}=(G,s,\{D_i\}_{i=1}^\ell,\{\cB_i\}_{i=1}^\ell)$ of the multiple key-cast problem such that for $i \in [k]$, $\cB_i=\{{\tt In}(v) \mid v \in V \setminus (D_i\cup \{s\})\}$.
Then $\bR(\cI) \geq 1$ if, for every terminal $d \in \cup_i D_i$, there exist two vertex-disjoint paths from $s$ to $d$, and, for every non-terminal node  $v$, there exist two edge-disjoint paths from $s$ to $v$.
Moreover, the combinatorial conditions are tight in the sense that there exist instances  ${\mathcal I}=(G,s,\{D_i\}_{i=1}^\ell,\{\cB_i\}_{i=1}^\ell)$  satisfying the conditions for which $\bR(\cI) = 1$.
\end{theorem}

We present the proof of Theorem~\ref{the:skd} below. 
In Section~\ref{sec:compare}, we compare the achievable key rate of our scheme with the maximum key-rate obtainable through source reconstruction and show a significant gap. Namely, we prove the following theorem.

\begin{theorem}[Secure multiple key-cast with source reconstruction]
\label{the:gap2}
Let $\epsilon >0$.
There exists an instance $\cI$ of the secure multiple key-cast problem that satisfies the combinatorial conditions of Theorem~\ref{the:skd} for which $\bRs(\cI) \leq 3/4 + \epsilon$.
\end{theorem}

\subsection{Proof of Theorem~\ref{the:skd}}
The proof is inspired by and closely follows the distributed secret sharing scheme presented in \cite{shah2013secure,shah2015distributed} for the threshold $k=1$.
The suggested dissemination scheme uses a special graph coloring of the vertices in the acyclic graph $G$.

We start be defining  the graph coloring, which assigns an integer color $c_v$ to each vertex $v \in V$.
Our coloring is designed to ensure that two vertices $u$ and $v$ have distinct colors ($c_u \ne c_v$) if and only if there exist directed paths   $P(s,u)$ from $s$ to $u$ and $P(s,v)$ from $s$ to $v$ that are vertex disjoint. Here, paths $P(s,u)$ and $P(s,v)$ are vertex disjoint if the only vertex that appears in both paths is the source $s$.
 
Our coloring proceeds in the predefined topological order.
Here and below, we assume that colors are assigned in increasing linear order, i.e., each time a distinct color is assigned, it's value is one larger than the previously assigned color.
The source $s$ receives color $c_s=1$.
Each neighbor $u$ of $s$ that only has incoming edges from $s$ is assigned a unique color.
For each subsequent vertex $u$, assume, by induction, that all vertices $v$ of topological order proceeding that of $u$ have been colored. 
If $u$ has two incoming edges $(v,u)$ and $(v',u)$ such that $c_v\ne c_{v'}$, then assign a color to $u$ that is distinct from all colors previously assigned. Such a vertex is referred to as a {\em newly colored} vertex.
Otherwise, $c_u$ takes the color of its incoming neighbors, i.e., $c_u=c_v$ for (any) incoming edge $(v,u)$.
Such vertices are called {\em color preserving}.
Each neighbor $u$ of $s$ that only has incoming edges from $s$ is called color preserving (despite the fact it is assigned a distinct color in the start of the procedure).
In Claim~\ref{claim:2vc} of Appendix~\ref{sec:app_sec}, we show that a  vertex $u$ in $G$ is 2-vertex connected from $s$ (i.e., in $G$ there are two vertex-disjoint paths $P_1(s,u)$ and $P_2(s,u)$) if and only if it is newly colored. 

Claim~\ref{claim:2vc} implies that every terminal node is newly colored.
As shown below, the color of each terminal determines its key, and keys of different colors are pairwise independent.
To allow terminals in the same terminal set to decode the same key, we slightly modify the coloring scheme.
Specifically, we pick, for each terminal set $D_i$, a representative terminal $d_i \in D_i$, and we assign all terminals in $D_i$ the color $c_{d_i}$.
Thus, the color representing terminal set $D_i$ is $c_{d_i}$.
As we assume in this work that terminal nodes do not have any outgoing edges, 
the suggested modified coloring does not change the color of nodes incoming to any network node $u$.


We now present the blocklength-$n$ key distribution scheme.
Assume the graph $G$ is colored by (a subset of) colors $\{1,2,3,\dots,c\}$.
We take $n$ to be sufficiently large such that $2^n > c$.
Consider the finite field $F=[2^n]$.
The source $s$ picks three independent values $s$, $a$, $b$ uniformly at random from $F$. 
For each neighbor $v$ of $s$ that only has incoming edges from $s$, the source transmits $s+c_va$ and $a+c_vb$ to $v$ (all operations are done over $F$).
This is possible, since any node in $G$ is 2-edge connected from the source (i.e., in this case, there are two edges connecting $s$ to $v$).
We proceed by topological order and show by induction that every vertex $u$  receives what it needs to compute $s+c_ua$ and $a+c_ub$.
Consider a network node $u$ (that may also be a terminal node).
If $u$ is newly colored, then it has at least 2 incoming edges $(v,u)$ and $(v',u)$ with $c_v \neq c_{v'}$.
In this case, $v$ transmits $(s+c_va)+c_u(a+c_vb)$ on $(v,u)$ and  $v'$ transmits $(s+c_{v'}a)+c_u(a+c_{v'}b)$ on $(v',u)$.
Rearranging the terms in the linear equations above, we conclude that $u$ receives $(s+c_ua) + c_v(a+c_ub)$ and $(s+c_u a) + c_{v'}(a+c_ub)$, which (as $c_v \ne c_{v'}$) allows it to decode $s+c_ua$ and $a+c_ub$.
If $u$ is not newly colored, then $u$ is a color preserving node.
Recall that any node, including node $u$, must have at least two incoming edges (otherwise it would not be two edge or vertex connected from $s$).
Let $(v,u)$ and $(v',u)$ be two incoming edges for $u$. 
Here $v$ may equal $v'$, and since $u$ is color preserving, 
$c_u=c_v=c_{v'}$.
Thus $v$ can forward  $s+c_ua$ on $(v,u)$ and $v'$ can forward $a+c_ub$ on $(v',u)$.

At the end of this process, every vertex $u$ in the graph $G$ has received exactly two distinct messages,  $s+c_ua$ and $a+c_ub$ (or two independent linear combinations thereof). 
For each $i \in [\ell]$, define the key for terminal set $D_i$ to be $K_i=s+c_{d_i}a$.
We conclude that every terminal $d$ in $D_i$ can recover $K_i$ after the protocol is complete.

To prove secrecy, we now use the fact that every terminal node is newly colored; that is,   
for $d \in D_i$, the color $c_{d_i}$ differs from $c_v$ for any vertex $v \not \in D_i$.
As any such $v$ only receives $(s+c_{v}a)$ and $(a+c_{v}b)$ (or linear combinations thereof) during the protocol, it holds that the mutual information between $v$'s messages and $K_i$ is zero. 
Formally, for any $i \in [\ell]$, if $v \not\in D_i \cup \{s\}$, then $c_{d_i} \ne c_v$ and thus
\begin{align*}
I(X_{{\tt In}(v)};K_i)= I((s+c_va),(a+c_vb);s+c_{d_i}a)=0.
\end{align*}

To show that the bound $\bR(\cI) \geq 1$ is tight under the combinatorial conditions assumed in Theorem~\ref{the:skd}, we now present an example instance $\cI$ (depicted in Figure~\ref{fig:skd}) that satisfies the conditions for which $\bR(\cI) = 1$.
\begin{figure}[t]
\begin{center}
\vspace{-1mm}
\includegraphics[width=2\columnwidth]{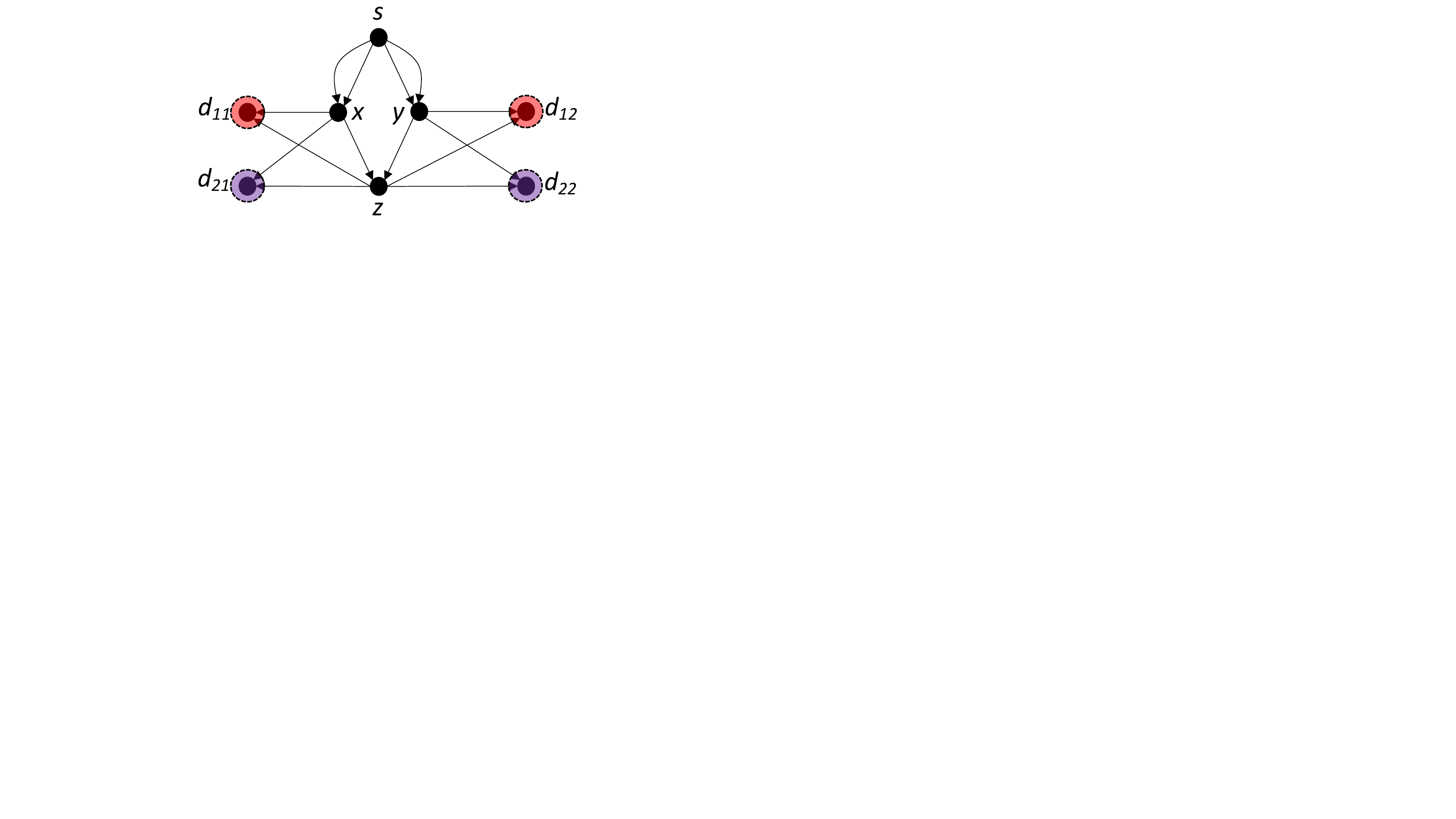}
\vspace{-75mm}
\caption{A tight example for Theorem~\ref{the:skd}.}
\label{fig:skd}
\end{center}
\end{figure}
In our example, the single source $s$ must disseminate two keys $K_1$ and $K_2$ to terminal sets $D_1=\{d_{11},d_{12}\}$ (in red) and $D_2=\{d_{21},d_{22}\}$ (in purple), respectively, such that (i) for each node $v \in \{x,y,z\}$ and for any $i\in\{1,2\}$, $I(K_i;X_{{\tt In}(v)})=0$, and (ii) for each $i,j \in \{1,2\}$, $i \ne j$, and any terminal node $d \in D_j$, $I(K_i;X_{{\tt In}(d)})=0$.
In Figure~\ref{fig:skd}, all edges have capacity 1.
Note that each terminal has two vertex-disjoint paths from $s$, and all nodes are two edge-connected from $s$.
We show that the maximum achievable key rate in this case is 1.


Consider any secret dissemination protocol.
For any vertex $v$, let $X_{{\tt In}(v)}$ be the incoming information to $v$ during the protocol, and for any edge $(u,v)$ let $X_{uv}$ be the information transmitted on $(u,v)$.
Here, as the network is acyclic, we consider communication according to topological order on $G$.
We next present a number of information inequalities that we  use to prove our assertion.

First consider edges $(x,z)$ and $(y,z)$, and note that $H(X_{{\tt In}(z)}) \leq H(X_{xz})+H(X_{yz})$.
Without loss of generality, let $H(X_{xz}) \geq H(X_{yz})$.
Then 
$H(X_{xz}) \geq 0.5 H(X_{{\tt In}(z)})$
and 
$H(X_{yz}|X_{xz}) \leq 0.5 H(X_{{\tt In}(z)})$.
%

Moreover, we now have that
\begin{align*}
H(X_{{\tt In}(x)},X_{{\tt In}(z)}) & = H(X_{{\tt In}(x)})+H(X_{{\tt In}(z)}|X_{{\tt In}(x)}) \\
& \leq 2 + H(X_{xz},X_{yz}|X_{{\tt In}(x)}) \\
& \leq 2+ H(X_{xz},X_{yz}|X_{xz})\\
& = 2+ H(X_{yz}|X_{xz})\\
& \leq 2 + 0.5H(X_{{\tt In}(z)}).
\end{align*}
Given our security requirements and graph topology, we know that $X_{x,d_{11}}$ is independent of $K_1$, that $X_{z,d_{11}}$ is independent of $K_1$, and that $K_1$ is a function of $X_{x,d_{11}}$ and $X_{z,d_{11}}$.
These observations imply that 
\begin{align*}
H(K_1) & = H(K_1|X_{z,d_{11}}) \leq  H(K_1,X_{x,d_{11}}|X_{z,d_{11}})\\
& = H(X_{x,d_{11}}|X_{z,d_{11}})+H(K_1|X_{x,d_{11}}X_{z,d_{11}}) \\
& = H(X_{x,d_{11}}|X_{z,d_{11}}).
\end{align*}
Similarly $H(K_1) \leq H(X_{z,d_{11}}|X_{x,d_{11}})$, and thus
\begin{align*}
H(X_{z,d_{11}},&X_{x,d_{11}})\geq \\
& H(X_{x,d_{11}}|X_{z,d_{11}})+H(X_{z,d_{11}}|X_{x,d_{11}}) \geq 2H(K_1).
\end{align*}
This, together with our security assumption that the incoming information to $d_{11}$ is independent of $K_2$, now implies that 
\begin{align*}
H(X_{{\tt In}(x)},X_{{\tt In}(z)}) & \geq H(X_{x,d_{11}},X_{z,d_{11}},K_2) \\
& = H(X_{x,d_{11}},X_{z,d_{11}})+H(K_2) \\
& \geq 2H(K_1)+H(K_2).
\end{align*}
Denote the key rate $H(K_1)=H(K_2)$ by $R$.
Finally, using the fact that $X_{{\tt In}(z)}$ is independent of $K_1$ we have
\begin{align*}
H(X_{{\tt In}(x)},X_{{\tt In}(z)}) & = H(X_{{\tt In}(x)},X_{{\tt In}(z)},K_1) \geq  H(X_{{\tt In}(z)},K_1)\\
& = H(X_{{\tt In}(z)})+H(K_1) = H(X_{{\tt In}(z)})+R.
\end{align*}
We now have the following inequalities
\begin{itemize}
\item $H(X_{{\tt In}(x)},X_{{\tt In}(z)}) \geq 3R$
\item $H(X_{{\tt In}(x)},X_{{\tt In}(z)}) \leq 2 + 0.5H(X_{{\tt In}(z)})$
\item $H(X_{{\tt In}(x)},X_{{\tt In}(z)}) \geq R+H(X_{{\tt In}(z)})$
\end{itemize}
The second and third inequalities above together imply that $2H(X_{{\tt In}(x)},X_{{\tt In}(z)}) \leq 4 + H(X_{{\tt In}(z)}) \leq 4 + H(X_{{\tt In}(x)},X_{{\tt In}(z)})-R$, implying that $H(X_{{\tt In}(x)},X_{{\tt In}(z)}) \leq  4 -R$. Combining this with the first inequality gives $3R \leq H(X_{{\tt In}(x)},X_{{\tt In}(z)}) \leq 4-R$ which proves our assertion that $R \leq 1$.

\section{Limitations of requiring source reconstruction}
\label{sec:compare}

In this section we show that in both the non-secure and secure settings the requirement for source reconstruction can significantly reduce the achievable key-rate when compared to that achievable without requiring source reconstruction.
We first present a technical lemma similar in nature to the Plotkin bound \cite{plotkin1960binary}, which states that {\em large} binary codes must have pairs of codewords with {\em small} total support. The lemma is proven in Appendix~\ref{sec:app_compare}.
\begin{lemma}
\label{lem:plotkin}
Any size-$M$, blocklength-$n$ binary code in which codewords are limited to Hamming weight $wn$ contains a pair of codewords $x=(x_1,\dots,x_n)$ and $x'=(x'_1,\dots,x'_n)$ such that the union of their support (i.e., the set $\{i \in [n] \mid x_i =1\} \cup \{i \in [n] \mid x'_i=1\}$) is of size at most $nw(2-w)\cdot\left(1+\frac{1}{M-1}\right)$.
\end{lemma}

In Appendix~\ref{sec:app_compare}, we prove Theorems~\ref{the:gap1} and \ref{the:gap2}, using the following corollary of Lemma~\ref{lem:plotkin} obtained by setting $M-1=\frac{1}{\epsilon} \geq \frac{w(2-w)}{\epsilon}$.
\begin{corollary}
\label{cor:plotkin}
Let $\epsilon>0$.
Any blocklength $n$ binary code of size $M = 1+\frac{1}{\epsilon}$ in which codewords are limited to Hamming weight $wn$ contains a pair of codewords $x=(x_1,\dots,x_n)$ and $x'=(x'_1,\dots,x'_n)$ such that the union of their support (i.e., the set $\{i \in [n] \mid x_i =1\} \cup \{i \in [n] \mid x'_i=1\}$) is of size at most $n(2w-w^2)+\epsilon n$.
\end{corollary}

\section{Concluding remarks}

In this work, we study the multiple key-cast problem in both the secure and non-secure settings.
For both settings, we present combinatorial conditions that allow multiple key-cast at unit rate. 
In the non-secure setting, our conditions are tight and characterize the key-rate.
In the secure case, we show that the analysis is tight in the sense that there exist instances satisfying the combinatorial conditions for which unit key-rate is optimal.
Our model assumes acyclic graphs with edge capacities that are integer multiples of the studied (unit) key-rate; both assumptions are used in the combinatorial coloring schemes and their analysis. 
Efforts to extend the analysis to cyclic graphs using, e.g., ideas from \cite{wang2007intersession,wang2010pairwise,shah2013secure,shah2015distributed}, or to general edge capacities are ongoing.
Our model assumes the distribution of a pair-wise independent collection of keys and considers, in the secure setting, a limited eavesdropper that controls a single network node.
Consideration of other forms of independence beyond pair-wise independence (e.g., $k$-wise independence)  under stronger eavesdropping models is also a subject of ongoing studies.


\appendix

\subsection{Proof of Theorem~\ref{the:kd}}
\label{sec:app_non}

The proof is inspired by \cite{wang2007intersession,wang2010pairwise,shenvi2010simple}, which address coding solutions and converses for $2$-unicast network coding with integral edge capacities. We start with the converse. 
Consider any key dissemination protocol of rate $R \geq 1$.
For any vertex $v$, let $X_{\tt In(v)}$ be the incoming information to $v$ during the protocol.
For any edge $e$, let $X_{e}$ be the information transmitted on $e$.
For any edge set $A$ let $X_{A}=(X_e: e \in A)$ be the information transmitted on edges $e \in A$.
We first note that for any edge $e$ in $C_j$ it must hold that $H(K_j|X_e)=H(X_e|K_j)=0$;
this follows since $e=e_d$ for some terminal $d \in D_j$ that requires key $K_j$ of rate $R \geq 1$ and $e$ has unit capacity.
It follows that $H(K_j|X_{C_j})=H(X_{C_j}|K_j)=0$.
Assume now, in contradiction, that there exist $i, j \in [\ell]$, $j \ne i$, and a terminal $d \in D_i$ such that $C_j$ separates $s$ from $d$. 
This implies, by our decoding requirements, that $H(K_i|X_{C_j})=0$.
However, as $H(X_{C_j}|K_j)=0$, we conclude that $H(K_i|K_j)=0$.
Hence, by the pairwise independence requirements, $R=H(K_i)=H(K_i|K_j)=0$, which contradicts our assumption that $R \geq 1$.

For achievability, we design a two-stage encoding scheme; both stages are deterministic.
Parts of the presentation below are repeated from Section~\ref{sec:results} for completeness.
First, we design a 2-multicast coding solution using a certain edge-coloring of $G$. 
Then, the coloring and coding scheme are modified to match our key dissemination requirements. 

Let the source $s$ hold 2 messages, $a$ and $b$. 
In our edge-colorings, an edge $e$ colored by the color $\alpha$ represents the transmission of the linear combination $a+\alpha b$ on $e$, where $a$, $b$, and $\alpha$ are all elements of a sufficiently large field $F=[2^n]$ for blocklength $n$, and all operations are over $F$.
Our coloring is governed by the predetermined topological order of edges in $G$.
We assume, without loss of generality, that every node in $G$ is connected from $s$.
Otherwise, one can remove such nodes from $G$ without impacting the communication protocol.

\noindent
{\bf $\bullet$ The first coloring stage:}
Consider the edge $e$ of least topological order.
Let $T_e$ be the set of edges that are disconnected from $s$ by the removal of $e$; we call such edges $e$-{\em tight}.
We color $e$ and every $e' \in T_e$ with the color $\alpha=1$ corresponding to the message $a+\alpha b = a+b$.
Notice that the coding scheme that transmits $a+\alpha b$ for $\alpha=1$ on $e$ and on all edges stemming from $e$ in $T_e$ is a valid key code in the sense that the incoming information to any edge suffices to compute its outgoing information.
Next, we continue coloring by induction over the topological order of edges $e$ in $G$.
In each step, we consider the next uncolored edge $e$ in topological order.
We color $e$ and the corresponding set $T_e$ (of edges disconnected from $s$ by the removal of $e$) by a new color $\alpha$, greater (by one) than all previous colors assigned; color $\alpha$ corresponds to the message $a + \alpha b$ to be communicated on $e$. 
Below, we prove that in any intermediate phase of our induction, any edge that has been assigned a color suffices to compute its outgoing information from its incoming information.
Our first coloring stage is depicted in Figure~\ref{fig:jj}(a). 

\noindent
{\bf $\bullet$ Validity of the encoding corresponding to the first coloring stage:}
Assume a partial coloring of the edges of $G$, and let $e$ be the uncolored edge with minimum topological order.
Let $\alpha$ be the distinct color that we now assign to $e$ and the $e$-tight edges in $T_e$.
We first show that all edges in $T_e$ are previously uncolored. 
Assume, otherwise, that there is a colored edge $e' \in T_e$.
This implies that the topological order of $e'$ is greater than that of $e$.
This, in turn, implies that $e' \in T_{e^*}$ for an edge $e^*$ that was previously assigned a color, or equivalently, that $e'$ is disconnected from $s$ by the removal of $e^*$.
Thus, every path from $s$ to $e'$ must first pass through $e^*$ and then through $e$.
However, as $e \not \in T_{e^*}$, there is a path connecting $s$ and $e$ that does not pass through $e^*$.
This implyies a path from $s$ to $e'$ that does not pass through $e^*$, in contradiction to $e' \in T_{e^*}$.

We now show that $e$ can compute its outgoing message of $a+\alpha b$ given its incoming information.
If $e$ is an outgoing edge of $s$, then the outgoing information on $e$ can be computed by $s$ as $s$ holds both $a$ and $b$.
Otherwise, note that all incoming edges of $e$ have been assigned colors.
It cannot be the case that $e$ has one incoming edge $e'$ with color $\alpha' < \alpha$ as otherwise the removal of $e'$ would have disconnected $e$ from $s$ and thus $e$ would have been colored in a previous stage of the inductive process.
It also cannot be the case that $e$ has more than one incoming edge and that all incoming edges $e'$ of $e$ have the same color.
In that case, it would hold that all incoming edges $e'$ to $e$ are in the set $T_{e^*}$ for an edge $e^*$ previously colored by the inductive process;
this would imply that all such $e'$ are disconnected from $s$ by the removal of $e^*$, and thus $e$ itself is in $T_{e^*}$.
If $e$ were in $T_{e^*}$, $e$ would have been colored in a previous stage of the inductive process.
We are left with the case that $e$ has two incoming edges with different colors.
In this case, as the information on these edges is independent, the tail vertex of $e$ can compute the outgoing message $a+ \alpha b$.

This concludes the first stage of our coloring/coding process.
After this first coloring, any edge $e=(u,v)$ that is 2-edge connected from $s$ (i.e., for which there exist at least two edge-disjoint paths connecting $s$ and $u$) cannot be in a set $T_{e^*}$ for any other edge $e^*$;  otherwise, by our definitions, edge $e$ is disconnected from $s$ by the removal of the edge $e^*$, and edge $e$ is not 2-edge connected from $s$.
The same holds for outgoing edges of $s$.
Thus, outgoing edges of $s$ and edges $e$ that are 2-edge connected from $s$ must have distinct colors.

\noindent
{\bf $\bullet$ The second coloring stage:} To initiate our second coloring/coding stage, we now focus on the cut sets $C_j$, defined previously, and on the set of edges $e'$ that are disconnected from $s$ by the removal of $C_j$. 
We denote this latter set of edges by $T_j$, and refer to such edges as $j$-{\em tight}.
We later prove that any edge can be $j$-tight for at most one value of $j \in [\ell]$.
By the topological-minimality condition in the definition of edges $e \in C_j$ (Definition~\ref{def:cutsets}), it holds that $e$ is either 2-edge connected from $s$ or is an outgoing edge of $s$.
By the above discussion, this implies for any $j$ and $j'$ ($j$ may equal $j'$),  that distinct edges $e \in C_j$ and $e' \in C_{j'}$ must have distinct colors.
In our second stage of coloring, for each $j \in [\ell]$, we assign a new color $\alpha_j$ to all edges $e$ in $C_j$ and to all $j$-tight edges $e' \in T_j$.

\noindent
{\bf $\bullet$ The sets $T_j$ and $T_{j'}$ for $j \ne j'$ are disjoint:} 
Before we discuss the validity of the coding scheme corresponding to the modified coloring of the second phase, we first show that any edge can be $j$-tight for at most one value of $j \in [\ell]$.
The proof is depicted in Figure~\ref{fig:jj}(b).
Assume otherwise, and let $e=(u,v)$ be an edge that is both $j$ and $j'$ tight.
Consider any path $P$ from $s$ to $e$.
The path must intersect both $C_j$ and $C_{j'}$.
Assume (without loss of generality) that the edge in $P$ that is farthest from $s$ (i.e., of maximum topological order) and intersects $C_j \cup C_{j'}$ is from $C_j$ and denote this edge by $e_j=(u_j,v_j)$.
Denote (one of) the edges in $P \cap C_{j'}$ by $e_{j'}$.
As $e_j \in C_j$ there must be a terminal $d_j \in D_j$ that is disconnected from $s$ by the removal of $e_j$.
By our assumptions on $G$, there exists a path $P'$ connecting $d_j$ with $s$ that does not intersect $C_{j'}$. 
As $P'$ must pass through $e_j$, the path from $s$ to $e$ that first uses the portion of $P'$ connecting $s$ and $e_j$ and then uses the portion of $P$ from $e_j$ until $e$, connects $s$ with $e$ and does not include any edges from $C_{j'}$.
The existence of such a path contradicts the assumption that $e$ is $j'$-tight.

\begin{figure}[t]
\begin{center}
\vspace{-10mm}
\includegraphics[width=2\columnwidth]{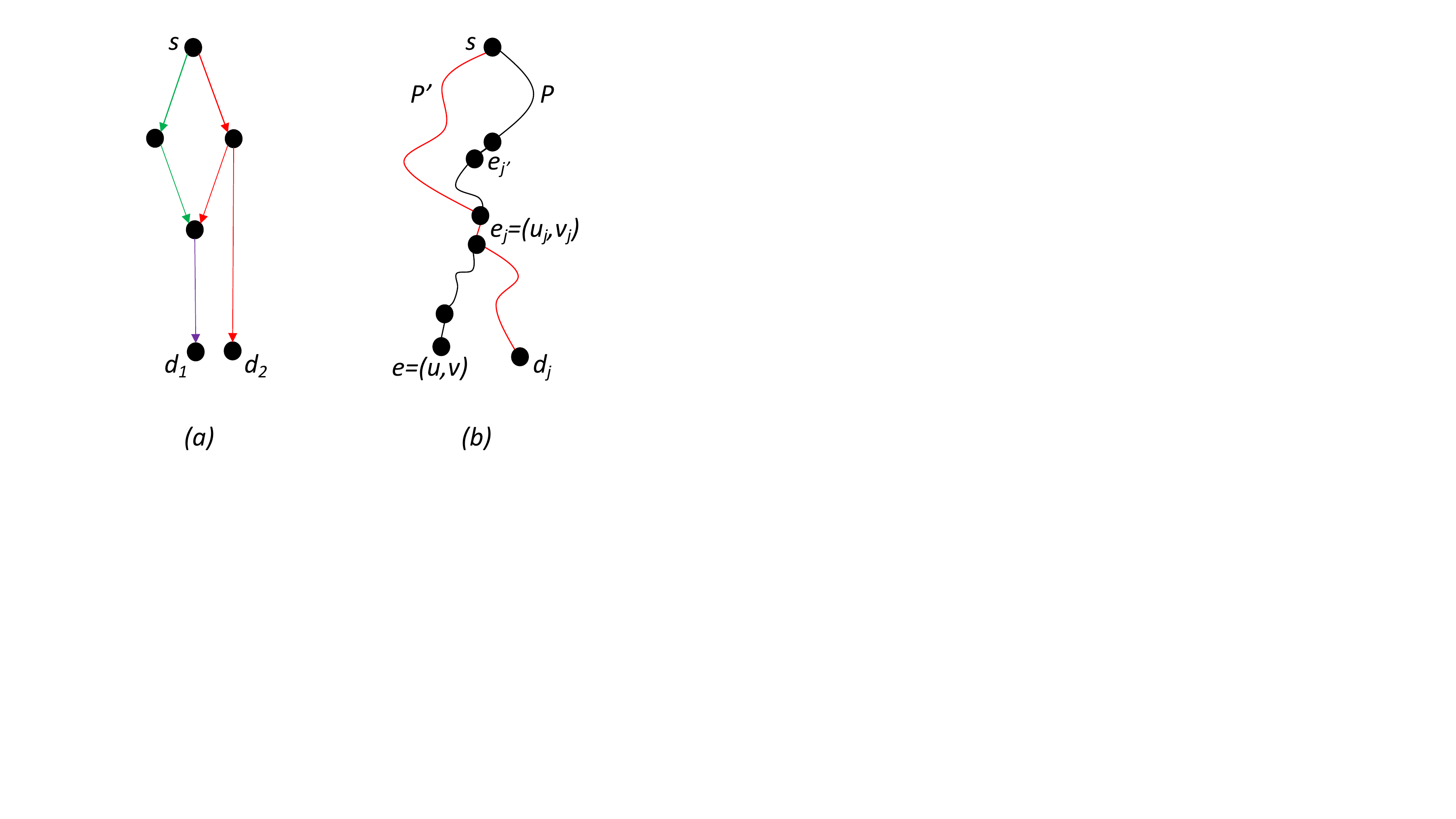}
\vspace{-48mm}
\caption{(a) A depiction of the first coloring phase in the proof of Theorem~\ref{the:kd}. (b) A depiction showing any edge can be $j$-tight for at most one value of $j \in [\ell]$ (from the proof of Theorem~\ref{the:kd}). The red path connects $s$ with $e$ without intersecting $C_{j'}$.}
\label{fig:jj}
\end{center}
\end{figure}

\noindent
{\bf $\bullet$ Validity of the encoding corresponding to the second coloring stage:} 
We next prove that our modified assignment of colors in the second coloring stage does not impact the network coding feasibility;
that is, we prove that, any edge $e$ can compute its outgoing information from its incoming information.
We proceed by the inductive order used in the first stage of coloring.
For the edge $e$ of minimum topological order,  $e$ is an outgoing edge of $s$; therefor, whether or not  $e$'s color has changed in the second stage of coloring, both $e$ and  all $e'$ in $T_e$ can  compute their outgoing information from their incoming information.
Below, we prove that in any intermediate phase of our induction, any edge that could previously  compute its outgoing information from its incoming information can do so after the modified coloring of the second stage.

Consider an edge $e$ and set $T_e$ that were assigned a new color in an intermediate step of the first coloring stage.
We consider several cases.

\noindent
{\bf $-$ The color of $e$ has changed between stages:}
If the color of $e$ has changed, then either $e \in C_j$ or $e \in T_j$.
If $e \in T_j$, then it is disconnected from $s$ by the removal of $C_j$, and thus all incoming edges $e'$ to $e$ must also be in $T_j$ or $C_j$. This implies that the incoming information to $e$ equals its outgoing information. 
If $e \in C_j$, then it is either an outgoing edge of $s$, in which case it can compute its outgoing information given the messages of $s$, or it is 2-edge connected from the source $s$, in which case it has at least 2 incoming edges.
Consider the incoming edges to $e$.
Some of these edges may have changed color in the second coloring stage while others may have preserved their original colors.
It cannot be the case that all incoming edges to $e$ changed color to $\alpha_{j'}$ for $j' \ne j$, as otherwise $e$ is disconnected from $s$ by the removal of $C_{j'}$, implying (by the definition of $C_j$) a terminal 
$d_j \in D_j$ that is disconnected from $s$ by the removal of $C_{j'}$.
This contradicts our assumptions on $G$.
It also cannot be the case that all incoming edges of $e$ have an identical color $\alpha$ which is unchanged from the previous stage of coloring.
This follows from the inductive analysis of the first coloring stage.
Namely, in such a case, edge $e$ is disconnected by the removal of $\alpha$-colored edges and thus would have been in $T_{e^*}$ for some edge $e^*$, implying that $e$ would not have received a new color in the first stage of coloring (in contradiction to our assumption on $e$).
Thus, $e$ must have 2 incoming edges with different colors (either two unchanged colors, one unchanged and one changed color, or two  colors that have been changed in the second stage).
This implies two incoming messages to $e$ which are independent, allowing the tail vertex of $e$ to compute the outgoing information on $e$.

\noindent
{\bf $-$ The color of $e$ did not change between stages:}
Finally, we consider an edge $e$ that did not change color between the first and second phase.
This case is similar to the analysis above.
If $e$ is an outgoing edge of $s$, then it can compute its outgoing information given the messages of $s$.
Otherwise, using the analysis of the first coloring phase, it cannot be the case that $e$ has only one incoming edge $e'$.
Consider the incoming edges to $e$.
Some of these edges may have changed color in the second coloring stage and some may have preserved their original colors.
It cannot be the case that all incoming edges to $e$ changed color to $\alpha_{j}$, as otherwise $e$ is disconnected from $s$ by the removal of $C_{j}$, implying that $e \in T_{j}$; this gives a contradiction since if $e \in T_j$ then $e$ would have changed color between coloring stages.
It also cannot be the case that all incoming edges of $e$ have the same color $\alpha$ which is unchanged by the second stage coloring. This follows from the inductive analysis of the first coloring stage.
Thus, again, $e$ must have 2 incoming edges with different colors (either two unchanged colors, one unchanged and one changed color, or two colors that have been changed in the second stage).
This implies two incoming messages to $e$ that are independent, allowing the tail vertex of $e$ to compute the outgoing information on $e$.
This concludes the analysis of the second stage of our coloring/coding process.

\noindent
{\bf $\bullet$ The decoding of $K_j=a + \alpha_j b$ at terminals $d \in D_j$:}
To finish our proof, we need to show that for any $j \in [\ell]$, any terminal $d \in D_j$ is able to decode $K_j = a+\alpha_j b$ (of rate 1).
Notice that this collection of keys is pair-wise independent.
As described in the body of this work, we here assume, without loss of generality, that all terminal nodes $d$ in $G$ have only one incoming edge. 
With this assumption, the single incoming edge $e$ to $d \in D_j$ is either in $C_j$ or in $T_j$.
This follows from the observation that $d$ is separated from $s$ by the removal of its single incoming edge, and thus there exists an edge $e_{d} \in C_j$ of minimum topological order disconnecting $d$ from $s$.
As edges in $C_j$ and $T_j$ are colored by $\alpha_j$,  terminal $d \in D_j$ can decode $K_j = a+\alpha_j b$.
This concludes our achievability proof.

\subsection{Claim~\ref{claim:2vc} used in the proof of Theorem~\ref{the:skd}}
\label{sec:app_sec}

\begin{claim}
\label{claim:2vc}
A vertex $u$ in $G$ is 2-vertex-connected from $s$ if and only if it is newly colored.
\end{claim}

\begin{proof}
For the forward direction, assume in contradiction that there is a 2-vertex-connected vertex $u$ for which all incoming edges $(v,u)$ have $c_v=c_0$ for a given color $c_0$.
Let $v^*$ be the vertex in $G$ with least topological order such that $c_{v^*}=c_0$.
Notice, by our coloring procedure, that for every vertex $v' \ne v^*$ in the graph $G$, if $c_{v'}=c_{v^*}=c_0$ then it must be the case that all incoming edges $(w,v')$ to $v'$ satisfy $c_{w}=c_0$.
We now claim that removing $c_{v^*}$ disconnects $u$ from $s$, in contradiction with our assumption that $u$ is 2-vertex-connected from $s$.
Assume in contradiction that there is a path $P$ from $s$ to $u$ that does not pass through $v^*$.
Let $v'$ be the vertex on $P$ with minimum topological order for which $c_{v'}=c_0$.
As all incoming edges $(v,u)$ to $u$ satisfy $c_v=c_0$, the vertex $v'$ is well defined.
The vertex $v'$ cannot be $v^*$  by our assumption that $P$ does not pass through $v^*$.
The vertex $v'$ cannot be $s$ since no other vertex except $s$ has color 1.
Thus the incoming edge $(w,v')$ to $v'$ along the path $P$ satisfies $c_{w}=c_0$; this contradicts the minimality assumption on the topological order of $v'$.

For the reverse direction, assume $u$ is not 2-vertex-connected. This implies that there is a single vertex $v$ in $G$ whose removal will disconnect $s$ from $u$. Consider the cut partition $(V_s,V_u)$ of $V$ implied by the removal of $v$ where $s \in V_s$ and $u \in V_u$.
It now follows by induction on the topological order of $G$ that all vertices in $V_u$ (including $u$) are color preserving with color $c_v$.
\end{proof}

\subsection{Proof of Lemma~\ref{lem:plotkin} and Theorems~\ref{the:gap1} and~\ref{the:gap2}}
\label{sec:app_compare}

\subsubsection{Proof of Lemma~\ref{lem:plotkin}} For any two codewords $x$ and $x'$, let $\ell_{x,x'}= \left|\{i \in [n] \mid x_i =1\} \cup \{i \in [n] \mid x'_i=1\}\right|$.
Let $\ell=\min_{x,x'}\ell_{x,x'}$.
We would like to show that for codes of size $M$, it holds that $\ell$ is at most $n(2w-w^2)+\frac{nw(1-w)}{M-1}$.
Let $\ell^c_{x,x'} = n-\ell_{x,x'}=\left|\{i \in [n] \mid x_i =0\} \cap \{i \in [n] \mid x'_i=0\}\right|$ be the number of entries $i$ in which both $x_i$ and $x'_i$ equal 0; and let $\ell^c=n-\ell$.
 On one hand, 
 $$
 \sum_{x \neq x'}\ell^c_{x,x'} \leq {{M} \choose {2}}\ell^c
 $$
 On the other hand, if $M_i$ is the number of codewords $x$ for which $x_i=0$, then
 $$
 \sum_{x \neq x'}\ell^c_{x,x'} =  \sum_{i=1}^n{{M_i} \choose {2}}
 $$
Given the weight limitation of codewords, notice that $\sum_i{M_i} \geq M(1-w)n$.
Moreover,  under this constraint, the expression $\sum_{i}{{M_i} \choose {2}}$ is minimized when, for each $i$, $M_i=M(1-w)$.
We thus conclude that  
$$
{{M}\choose {2}}\ell^c \geq \sum_{x \neq x'}\ell^c_{x,x'} =  \sum_{i}{{M_i} \choose {2}} \geq  n{{M(1-w)} \choose {2}}
$$
Thus, $\ell^c \geq n(1-w)^2-\frac{nw(1-w)}{M-1}$, or, equivalently,  
$$
\ell \leq n-n(1-w)^2+\frac{nw(1-w)}{M-1} \leq nw(2-w)\cdot\left(1+\frac{1}{M-1}\right).
$$

\subsubsection{The non-secure case: proof of Theorem~\ref{the:gap1}}

We now prove Theorem~\ref{the:gap1} using the instance depicted in Figure~\ref{fig:gap1}. 

\begin{figure}[t]
\begin{center}
\includegraphics[width=2\columnwidth]{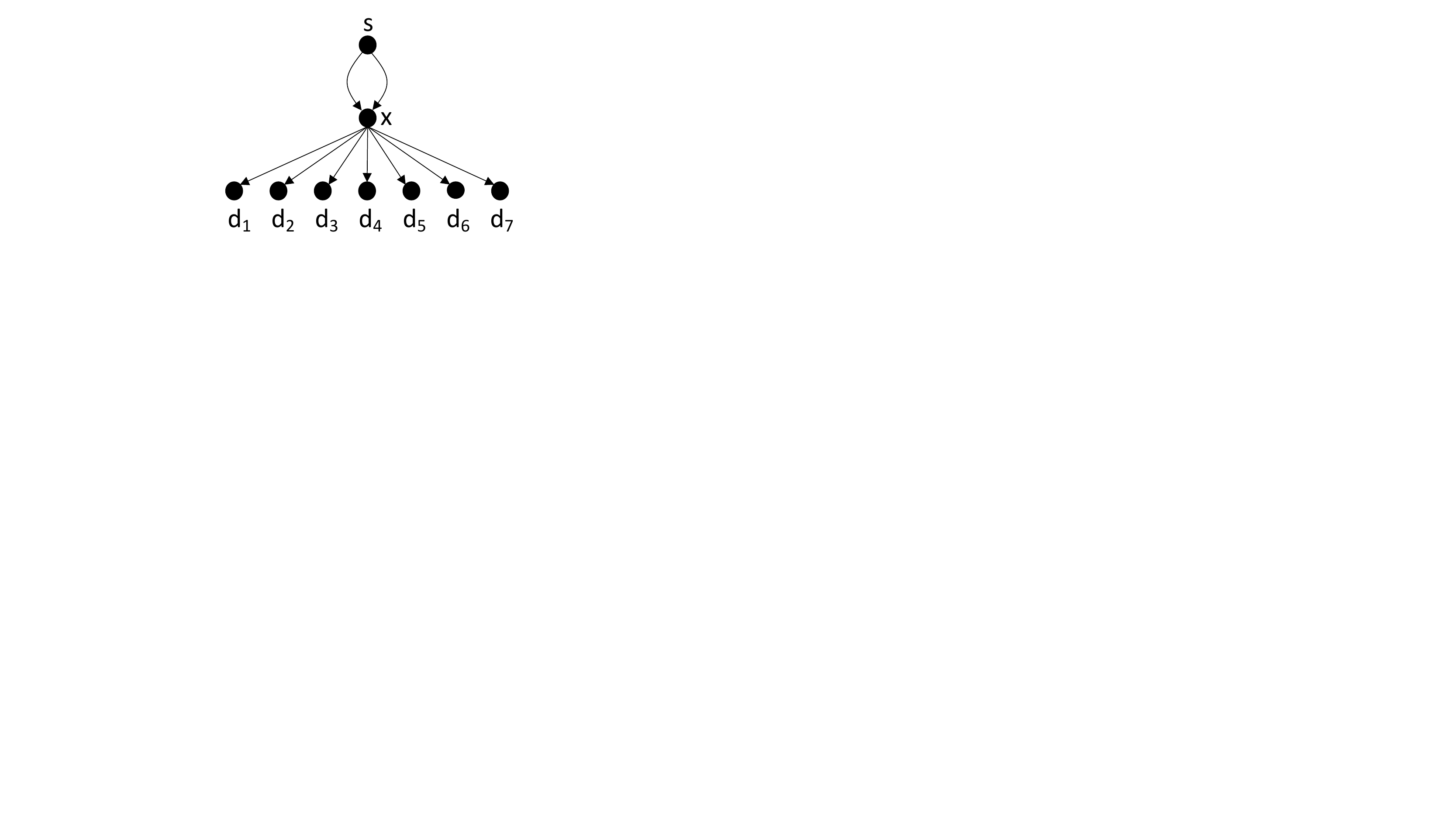}
\vspace{-70mm}
\caption{An example instance for Theorem~\ref{the:gap1}.}
\label{fig:gap1}
\end{center}
\end{figure}

\begin{proof} (of Theorem~\ref{the:gap1})
We consider the 3-layered instance depicted in Figure~\ref{fig:gap1}.
The network includes a source $s$ connected by 2 unit-capacity edges to an intermediate node $x$. The source $s$ and node $x$ represent the first two layers of the network. The third layer consists of terminal nodes $d_1,\dots,d_\ell$ for $\ell=1+1/\epsilon$, each connected with a unit capacity edge from $x$, and each belonging to a distinct terminal set $D_i=\{d_i\}$.
Note that this instance satisfies the conditions of Theorem~\ref{the:kd}.
Each terminal decodes a subset of the source information bits $M=\{b_1,b_2,\dots\}$; subset $M_i$ is decoded at terminal $d_i$ for $i \in [\ell]$.
Consider any blocklength-$n$  code over the network.
Let $M_x$ be the subset of source bits that are decodable at node $x$.
We first notice that $|M_x| \leq 2n$, and thus we assume, without loss of generality, that $M_x \subseteq \{b_1,\dots,b_{2n}\}$.
Given the graph topology, we also observe, for each $i \in [\ell]$,
that $M_{i}$ 
is a subset of $M_x$ with $|M_i| \leq n$.
Considering the characteristic binary vector $c_i$ of $M_i$ as a subset of $M_x$, we obtain a codebook $c_1,\dots,c_\ell$ of codewords each of blocklength (at most) $2n$ and of weight at most $n$.
Appending zeros to codewords if needed, we obtain a codebook $c_1,\dots,c_\ell$ of blocklength $2n$ and of weight at most $n$.
Applying Corollary~\ref{cor:plotkin} with $w=1/2$ and blocklength $2n$, we conclude that there exist indices $i \ne j$ such that the total support of $c_i$ and $c_j$, and that of $M_i$ and $M_j$, correspondingly, is at most $(2w-w^2+\epsilon)2n=(3+4\epsilon)n/2$.
Since $K_i$ is a function of $M_{i}$, $K_j$ is a function of $M_j$, and $K_i$ is independent of $K_j$, we conclude that $H(K_i)+H(K_j) \leq (3+4\epsilon)n/2$, which in turn implies that $\bRs(\cI) \leq (3+4\epsilon)/4=3/4 + \epsilon$.
\end{proof}

\begin{remark}
\label{rem:kd}
Given the connectivity conditions of Theorem~\ref{the:kd}, one can show, using random linear network coding over blocklength $n$, that every terminal node can decode two (uniformly distributed) messages, each of entropy $n/2$. Thus each terminal, using a potentially different linear combination of the decoded messages, can obtain a key of rate $1/2$ that is independent of any key decoded by a terminal in a different decoding set. This simple scheme implies that $\bRs(\cI) \geq 1/2$.
Thus, the gap presented in Theorem~\ref{the:gap1} between key dissemination schemes with and without source reconstruction, while not  necessarily optimal, is of the correct order. 
\end{remark}

\subsection{The secure case: proof of Theorem~\ref{the:gap2}}
We now show that the instance depicted in Figure~\ref{fig:skd}
(showing a tight example for Theorem~\ref{the:skd}) can be slightly modified to prove Theorem~\ref{the:gap2}.

\begin{proof} (of Theorem~\ref{the:gap2})
We start by defining the instance $\cI$, which is a modified version of that given in Figure~\ref{fig:skd}.
Let the number of terminal sets $D_i$ be $\ell = \frac{9}{\epsilon}\left(1+\frac{9}{\epsilon}\right)$ instead of $\ell=2$.
In $\cI$, for $i \in [\ell]$, the single source $s$ must disseminate key $K_i$ to terminal set $D_i=\{d_{i1},d_{i2}\}$ such that (i) for each node $v \in \{x,y,z\}$ and for any $i\in [\ell]$, $I(K_i;X_{{\tt In}(v)})=0$, and (ii) for any terminal node $d$, ${\rm In}(d)$ does not reveal any information about a key that is not required at $d$. There are two edges connecting $s$ and $x$ and two edges connecting $s$ and $y$.
All  edges in Figure~\ref{fig:skd} have unit capacity.
Note that each terminal has 2 vertex-disjoint paths from $s$, and all nodes are two-edge connected from $s$.
Below, we show that $\bRs(\cI) \leq 3/4 + \epsilon$.

For each $i \in [\ell]$, let $M_{i}$ be the source message bits reconstructed at terminal node $d_{i1}$.
Let $B=\cup_{i}M_{i}$.
For each $M_i$, we have $H(M_i|X_{{\tt In}(x)}X_{{\tt In}(z)})=0$.
Thus, it also holds that $H(B|X_{{\tt In}(x)}X_{{\tt In}(z)})=0$.
We conclude that, $H(B) \leq H(X_{{\tt In}(x)}X_{{\tt In}(z)}) \leq H(X_{{\tt In}(x)})+H(X_{{\tt In}(z)}|X_{{\tt In}(x)}) \leq 3$.

Considering the characteristic binary vector $c_i$ of $M_i$ as a subset of $B$, we obtain a codebook $c_1,\dots,c_\ell$ of codewords, each of blocklength $|B| \leq 3n$ and of weight at most $2n$.
Appending zeros to codewords if needed, we obtain a codebook $c_1,\dots,c_\ell$ of blocklength $3n$ and of weight at most $2n$.
By the pigeonhole principle, there exists a weight $w$ such that at least $1+\frac{9}{\epsilon}$ of the terminals $\{d_{i1}\}_i$ decode $M_i$ which is of size in the range $[3wn,3(w+\epsilon/9)n]$.
Applying Corollary~\ref{cor:plotkin} with such $w$ and blocklength $3n$, we conclude that  there exist indices $i \ne i'$ such that the total support of $c_i$ and $c_{i'}$, and that of $M_i$ and $M_{i'}$, are at most $(2w-w^2)3n+\epsilon n$.
As $K_i$ is a function of $M_{i}$, $K_{i'}$ is a function of $M_{i'}$, and $K_i$ is independent of $M_{i'}$ we conclude that 
\begin{align*}
H(K_i) & \leq |M_i \setminus M_{i'}| = |M_i \cup M_{i'}|-|M_{i'}| \\
& \leq (2w-w^2-w)3n+\epsilon n \\
& = (w-w^2)3n+\epsilon n \leq \frac{3n}{4}+\epsilon n.
\end{align*}
This implies that $\bRs(\cI) \leq 3/4 + \epsilon$.
\end{proof}

\begin{remark}
Given the connectivity conditions of Theorem~\ref{the:skd}, one can show using random linear network coding combined with additional ideas that  $\bRs(\cI) \geq 1/2$. Thus, as in Remark~\ref{rem:kd}, the gap presented in Theorem~\ref{the:gap2} between key dissemination schemes with and without source reconstruction, while also not necessarily optimal, is of the correct order. 
\end{remark}

\newpage
\thispagestyle{empty}

{\tiny{
\bibliographystyle{unsrt}
\bibliography{ref}

\begin{thebibliography}{10}

\bibitem{shannon1949communication}
Claude~E. Shannon.
\newblock Communication theory of secrecy systems.
\newblock {\em The Bell system technical journal}, 28(4):656--715, 1949.

\bibitem{maurer1993secret}
Ueli~M. Maurer.
\newblock Secret key agreement by public discussion from common information.
\newblock {\em IEEE Transactions on Information Theory}, 39(3):733--742, 1993.

\bibitem{ahlswede1993common}
Rudolf Ahlswede and Imre Csisz{\'a}r.
\newblock {Common randomness in information theory and cryptography. I. Secret
  sharing}.
\newblock {\em IEEE Transactions on Information Theory}, 39(4):1121--1132,
  1993.

\bibitem{lapidoth1998reliable}
Amos Lapidoth and Prakash Narayan.
\newblock Reliable communication under channel uncertainty.
\newblock {\em IEEE transactions on Information Theory}, 44(6):2148--2177,
  1998.

\bibitem{chen2021breaking}
Wei-Ning Chen, Peter Kairouz, and Ayfer Ozgur.
\newblock Breaking the dimension dependence in sparse distribution estimation
  under communication constraints.
\newblock In {\em Conference on Learning Theory}, pages 1028--1059. PMLR, 2021.

\bibitem{acharya2019communication}
Jayadev Acharya and Ziteng Sun.
\newblock Communication complexity in locally private distribution estimation
  and heavy hitters.
\newblock In {\em International Conference on Machine Learning}, pages 51--60.
  PMLR, 2019.

\bibitem{gacs1973common}
Peter G{\'a}cs and J{\'a}nos K{\"o}rner.
\newblock Common information is far less than mutual information.
\newblock {\em Problems of Control and Information Theory}, 2(2):149--162,
  1973.

\bibitem{wyner1975common}
Aaron~D. Wyner.
\newblock The common information of two dependent random variables.
\newblock {\em IEEE Transactions on Information Theory}, 21(2):163--179, 1975.

\bibitem{wyner1975wire}
Aaron~D. Wyner.
\newblock The wire-tap channel.
\newblock {\em Bell system technical journal}, 54(8):1355--1387, 1975.

\bibitem{witsenhausen1975sequences}
Hans~S. Witsenhausen.
\newblock On sequences of pairs of dependent random variables.
\newblock {\em SIAM Journal on Applied Mathematics}, 28(1):100--113, 1975.

\bibitem{csiszar1978broadcast}
Imre Csisz{\'a}r and Janos Korner.
\newblock Broadcast channels with confidential messages.
\newblock {\em IEEE Transactions on Information Theory}, 24(3):339--348, 1978.

\bibitem{bennett1988privacy}
Charles~H. Bennett, Gilles Brassard, and Jean-Marc Robert.
\newblock Privacy amplification by public discussion.
\newblock {\em SIAM Journal on Computing}, 17(2):210--229, 1988.

\bibitem{bennett1995generalized}
Charles~H. Bennett, Gilles Brassard, Claude Cr{\'e}peau, and Ueli~M. Maurer.
\newblock Generalized privacy amplification.
\newblock {\em IEEE Transactions on Information theory}, 41(6):1915--1923,
  1995.

\bibitem{maurer1997privacy}
Ueli Maurer and Stefan Wolf.
\newblock Privacy amplification secure against active adversaries.
\newblock In {\em Annual International Cryptology Conference}, pages 307--321.
  Springer, 1997.

\bibitem{ahlswede1998common}
Rudolf Ahlswede and Imre Csisz{\'a}r.
\newblock {Common randomness in information theory and cryptography. II. CR
  capacity}.
\newblock {\em IEEE Transactions on Information Theory}, 44(1):225--240, 1998.

\bibitem{csiszar2000common}
Imre Csisz{\'a}r and Prakash Narayan.
\newblock Common randomness and secret key generation with a helper.
\newblock {\em IEEE Transactions on Information Theory}, 46(2):344--366, 2000.

\bibitem{csiszar2004secrecy}
Imre Csisz{\'a}r and Prakash Narayan.
\newblock Secrecy capacities for multiple terminals.
\newblock {\em IEEE Transactions on Information Theory}, 50(12):3047--3061,
  2004.

\bibitem{mossel2006non}
Elchanan Mossel, Ryan O'Donnell, Oded Regev, Jeffrey~E. Steif, and Benny
  Sudakov.
\newblock {Non-interactive correlation distillation, inhomogeneous Markov
  chains, and the reverse Bonami-Beckner inequality}.
\newblock {\em Israel Journal of Mathematics}, 154(1):299--336, 2006.

\bibitem{chan2014multiterminal}
Chung Chan and Lizhong Zheng.
\newblock Multiterminal secret key agreement.
\newblock {\em IEEE Transactions on Information Theory}, 60(6):3379--3412,
  2014.

\bibitem{csiszar2008secrecy}
Imre Csisz{\'a}r and Prakash Narayan.
\newblock Secrecy capacities for multiterminal channel models.
\newblock {\em IEEE Transactions on Information Theory}, 54(6):2437--2452,
  2008.

\bibitem{gohari2010information}
Amin~Aminzadeh Gohari and Venkat Anantharam.
\newblock {Information-theoretic key agreement of multiple terminals—Part I}.
\newblock {\em IEEE Transactions on Information Theory}, 56(8):3973--3996,
  2010.

\bibitem{gohari2010information2}
Amin~Aminzadeh Gohari and Venkat Anantharam.
\newblock {Information-theoretic key agreement of multiple terminals—Part II:
  Channel model}.
\newblock {\em IEEE Transactions on Information Theory}, 56(8):3997--4010,
  2010.

\bibitem{siavoshani2010group}
Mahdi~Jafari Siavoshani, Christina Fragouli, Suhas Diggavi, Uday Pulleti, and
  Katerina Argyraki.
\newblock Group secret key generation over broadcast erasure channels.
\newblock In {\em Forty Fourth IEEE Asilomar Conference on Signals, Systems and
  Computers}, pages 719--723, 2010.

\bibitem{bogdanov2011extracting}
Andrej Bogdanov and Elchanan Mossel.
\newblock On extracting common random bits from correlated sources.
\newblock {\em IEEE Transactions on Information Theory}, 57(10):6351--6355,
  2011.

\bibitem{tyagi2013common}
Himanshu Tyagi.
\newblock Common information and secret key capacity.
\newblock {\em IEEE Transactions on Information Theory}, 59(9):5627--5640,
  2013.

\bibitem{chan2014extracting}
Siu~On Chan, Elchanan Mossel, and Joe Neeman.
\newblock On extracting common random bits from correlated sources on large
  alphabets.
\newblock {\em IEEE Transactions on Information Theory}, 60(3):1630--1637,
  2014.

\bibitem{liu2015secret}
Jingbo Liu, Paul Cuff, and Sergio Verd{\'u}.
\newblock Secret key generation with one communicator and a one-shot converse
  via hypercontractivity.
\newblock In {\em 2015 IEEE International Symposium on Information Theory
  (ISIT)}, pages 710--714, 2015.

\bibitem{guruswami2016tight}
Venkatesan Guruswami and Jaikumar Radhakrishnan.
\newblock Tight bounds for communication-assisted agreement distillation.
\newblock In {\em 31st Conference on Computational Complexity (CCC)}, 2016.

\bibitem{xu2016private}
Peng Xu, Zhiguo Ding, Xuchu Dai, and George~K. Karagiannidis.
\newblock On the private key capacity of the $ m $-relay pairwise independent
  network.
\newblock {\em IEEE Transactions on Information Theory}, 62(7):3831--3843,
  2016.

\bibitem{hayashi2016secret}
Masahito Hayashi, Himanshu Tyagi, and Shun Watanabe.
\newblock Secret key agreement: General capacity and second-order asymptotics.
\newblock {\em IEEE Transactions on Information Theory}, 62(7):3796--3810,
  2016.

\bibitem{narayan2016multiterminal}
Prakash Narayan and Himanshu Tyagi.
\newblock {\em Multiterminal secrecy by public discussion}.
\newblock Now Publishers, Hanover, MA, USA, 2016.

\bibitem{ghazi2018resource}
Badih Ghazi and T.~S. Jayram.
\newblock Resource-efficient common randomness and secret-key schemes.
\newblock In {\em Proceedings of the Twenty-Ninth Annual ACM-SIAM Symposium on
  Discrete Algorithms}, pages 1834--1853, 2018.

\bibitem{liu2016common}
Jingbo Liu, Paul Cuff, and Sergio Verd{\'u}.
\newblock Common randomness and key generation with limited interaction.
\newblock {\em arXiv preprint arXiv:1601.00899}, 2016.

\bibitem{canonne2017communication}
Cl{\'e}ment~L. Canonne, Venkatesan Guruswami, Raghu Meka, and Madhu Sudan.
\newblock Communication with imperfectly shared randomness.
\newblock {\em IEEE Transactions on Information Theory}, 63(10):6799--6818,
  2017.

\bibitem{LE:22}
Michael Langberg and Michelle Effros.
\newblock Network coding multicast key-capacity.
\newblock In {\em IEEE Information Theory Workshop (ITW)}, pages 422--427,
  2022.

\bibitem{cai2002secure}
Ning Cai and Raymond~W. Yeung.
\newblock Secure network coding.
\newblock In {\em IEEE International Symposium on Information Theory}, page
  323, 2002.

\bibitem{feldman2004capacity}
Jon Feldman, Tal Malkin, C.~Stein, and Rocco~A. Servedio.
\newblock On the capacity of secure network coding.
\newblock In {\em 42nd Annual Allerton Conference on Communication, Control,
  and Computing}, pages 63--68, 2004.

\bibitem{cai2007security}
Ning Cai and Raymond~W. Yeung.
\newblock A security condition for multi-source linear network coding.
\newblock {\em IEEE International Symposium on Information Theory}, pages
  561--565, 2007.

\bibitem{yeung2008optimality}
Ning Cai and Raymond~W. Yeung.
\newblock On the optimality of a construction of secure network codes.
\newblock {\em IEEE International Symposium on Information Theory}, pages
  166--170, 2008.

\bibitem{el2012secure}
Salim~El Rouayheb, Emina Soljanin, and Alex Sprintson.
\newblock {Secure network coding for wiretap networks of type II}.
\newblock {\em IEEE Transactions on Information Theory}, 58(3):1361--1371,
  2012.

\bibitem{silva2011universal}
Danilo Silva and Frank~R. Kschischang.
\newblock Universal secure network coding via rank-metric codes.
\newblock {\em IEEE Transactions on Information Theory}, 57(2):1124--1135,
  2011.

\bibitem{jaggi2012secure}
Sidharth Jaggi and Michael Langberg.
\newblock Secure network coding: Bounds and algorithms for secret and reliable
  communications.
\newblock In {\em Chapter 7 of Network Coding: Fundamentals and applications
  (Muriel M{\'e}dard and Alex Sprintson ed.)}, pages 183--215. Academic Press,
  2012.

\bibitem{lim2005extracting}
Daihyun Lim, Jae~W. Lee, Blaise Gassend, Edward~G. Suh, Marten Van~Dijk, and
  Srinivas Devadas.
\newblock Extracting secret keys from integrated circuits.
\newblock {\em IEEE Transactions on Very Large Scale Integration (VLSI)
  Systems}, 13(10):1200--1205, 2005.

\bibitem{su2008digital}
Ying Su, Jeremy Holleman, and Brian~P. Otis.
\newblock A digital 1.6 pj/bit chip identification circuit using process
  variations.
\newblock {\em IEEE Journal of Solid-State Circuits}, 43(1):69--77, 2008.

\bibitem{suh2007physical}
G.~Edward Suh and Srinivas Devadas.
\newblock Physical unclonable functions for device authentication and secret
  key generation.
\newblock In {\em 2007 44th ACM/IEEE Design Automation Conference}, pages
  9--14, 2007.

\bibitem{yu2009towards}
Haile Yu, Philip Heng~Wai Leong, Heiko Hinkelmann, Leandro Moller, Manfred
  Glesner, and Peter Zipf.
\newblock {Towards a unique FPGA-based identification circuit using process
  variations}.
\newblock In {\em International Conference on Field Programmable Logic and
  Applications}, pages 397--402, 2009.

\bibitem{chen2020breaking}
Wei-Ning Chen, Peter Kairouz, and Ayfer Ozgur.
\newblock Breaking the communication-privacy-accuracy trilemma.
\newblock {\em Advances in Neural Information Processing Systems},
  33:3312--3324, 2020.

\bibitem{byrd2020differentially}
David Byrd and Antigoni Polychroniadou.
\newblock Differentially private secure multi-party computation for federated
  learning in financial applications.
\newblock In {\em First ACM International Conference on AI in Finance}, pages
  1--9, 2020.

\bibitem{ahlswede2021identification}
Rudolf Ahlswede, Alexander Ahlswede, Ingo Alth{\"o}fer, Christian Deppe, and
  Ulrich Tamm.
\newblock {\em Identification and Other Probabilistic Models}.
\newblock Springer, 2021.

\bibitem{shamir1979share}
Adi Shamir.
\newblock How to share a secret.
\newblock {\em Communications of the ACM}, 22(11):612--613, 1979.

\bibitem{ito1989secret}
Mitsuru Ito, Akira Saito, and Takao Nishizeki.
\newblock Secret sharing scheme realizing general access structure.
\newblock {\em Electronics and Communications in Japan (Part III: Fundamental
  Electronic Science)}, 72(9):56--64, 1989.

\bibitem{beimel2011secret}
Amos Beimel.
\newblock Secret-sharing schemes: A survey.
\newblock In {\em International conference on coding and cryptology}, pages
  11--46. Springer, 2011.

\bibitem{shah2013secure}
Nihar~B. Shah, K.~V. Rashmi, and Kannan Ramchandran.
\newblock Secure network coding for distributed secret sharing with low
  communication cost.
\newblock In {\em IEEE International Symposium on Information Theory}, pages
  2404--2408, 2013.

\bibitem{shah2015distributed}
Nihar~B. Shah, K.~V. Rashmi, and Kannan Ramchandran.
\newblock Distributed secret dissemination across a network.
\newblock {\em IEEE Journal of Selected Topics in Signal Processing},
  9(7):1206--1216, 2015.

\bibitem{wang2007intersession}
Chih-Chun Wang and Ness~B. Shroff.
\newblock Intersession network coding for two simple multicast sessions.
\newblock In {\em 45th Annual Allerton Conference on Communication, Control,
  and Computing (Allerton)}, pages 682--689, 2007.

\bibitem{wang2010pairwise}
Chih-Chun Wang and Ness~B. Shroff.
\newblock Pairwise intersession network coding on directed networks.
\newblock {\em IEEE Transactions on Information Theory}, 56(8):3879--3900,
  2010.

\bibitem{shenvi2010simple}
Sagar Shenvi and Bikash~Kumar Dey.
\newblock A simple necessary and sufficient condition for the double unicast
  problem.
\newblock In {\em IEEE International Conference on Communications}, pages 1--5,
  2010.

\bibitem{fragouli2007network}
Christina Fragouli and Emina Soljanin.
\newblock Network coding fundamentals.
\newblock {\em Foundations and Trends{\textregistered} in Networking},
  2(1):1--133, 2007.

\bibitem{blakley1979safeguarding}
George~Robert Blakley.
\newblock Safeguarding cryptographic keys.
\newblock In {\em IEEE International Workshop on Managing Requirements
  Knowledge}, page 313, 1979.

\bibitem{plotkin1960binary}
Morris Plotkin.
\newblock Binary codes with specified minimum distance.
\newblock {\em IRE Transactions on Information Theory}, 6(4):445--450, 1960.

\end{thebibliography}
}}

\end{document}